\documentclass[runningheads,a4paper]{llncs}

\usepackage{amssymb}
\usepackage{graphicx}
\usepackage[fleqn]{amsmath}
\usepackage{esvect}
\usepackage[]{algorithm2e}

\usepackage{url}
\urldef{\mailsa}\path|rupei.xu@utdallas.edu|
\urldef{\mailsb}\path|wshull@emory.edu|
\newcommand{\keywords}[1]{\par\addvspace\baselineskip
\noindent\keywordname\enspace\ignorespaces#1}

\begin{document}

\mainmatter  

\title{Hedge Connectivity without Hedge Overlaps}

\titlerunning{Hedge Connectivity without Hedge Overlaps}

%
%
\author{ Rupei Xu \and Warren Shull}
\authorrunning{ Rupei Xu \and Warren Shull }

\institute{The University of Texas at Dallas\\
Emory University\\
\mailsa\\
\mailsb\\
}

%
%

\toctitle{Hedge Connectivity without Hedge Overlaps}
\tocauthor{ Rupei Xu, Warren Shull}
\maketitle

\begin{abstract}

Connectivity is a central notion of graph theory and plays an important role in graph algorithm design and applications. With emerging new applications in networks, a new type of graph connectivity problem has been getting more attention--hedge connectivity. In this paper, we consider the model of hedge graphs without hedge overlaps, where edges are partitioned into subsets called hedges that fail together. The hedge connectivity of a graph is the minimum number of hedges whose removal disconnects the graph. This model is more general than the hypergraph, which brings new computational challenges. It has been a long open problem whether this problem is solvable in polynomial time. In this paper, we study the combinatorial properties of hedge graph connectivity without hedge overlaps, based on its extremal conditions as well as hedge contraction operations, which provide new insights into its algorithmic progress. 

\end{abstract}

\section{Introduction}

Connectivity has been a central notion of graph theory since its birth in the 18th century and has been playing an important role in graph algorithm design and applications. With emerging real-world new applications in image segmentation, optical networking, network security, software-defined networking and virtual network embedding, a new type of graph connectivity problem has been getting more attention--hedge connectivity, where edges are partitioned into subsets called hedges that fail together. The hedge connectivity of a graph is the minimum number of hedges whose removal disconnects the graph. However, its mathematical studies are still very limited. This is the first paper to investigate the combinatorial properties of hedge graph connectivity without hedge overlaps, based on its extremal conditions as well as hedge contraction operations. 

In this paper, all original graphs considered are finite and simple, i.e., they have no self-loops nor multiple edges. But during operations, self-loops and multiple edges may appear. Graphs are also assumed connected, otherwise, the hedge connectivity is just simply zero. 

Given an undirected graph $G=(V, E, L)$, where $V=\{v_1, v_2, ..., v_n\}$ is the vertex set, $E=\{e_1, e_2, ..., e_m\}$ is the edge set and $L=\{\ell_1, \ell_2,..., \ell_{|L|}\}$is the label set. Each edge has one label from $L$, and the edge set with same label $L_i$ form an hedge $H_i.$ Let $span(H_i)$ represent the \emph{\textbf{span}} of hedge $H_i$, i.e., the number of its components. The \emph{\textbf{rank}} of a graph $rank(G)$ is the difference of the number of its vertices and its span: $rank(G)=|V(G)|-span(G).$ The \emph{\textbf{nullity}} of a graph $nullity$ is the difference of the number of its edges and its rank: $nullity(G)=|E(G)|-rank(G)=|E(G)|-|V(G)|+span(G).$

\begin{figure}
	\centering
	\includegraphics[width=0.8\linewidth]{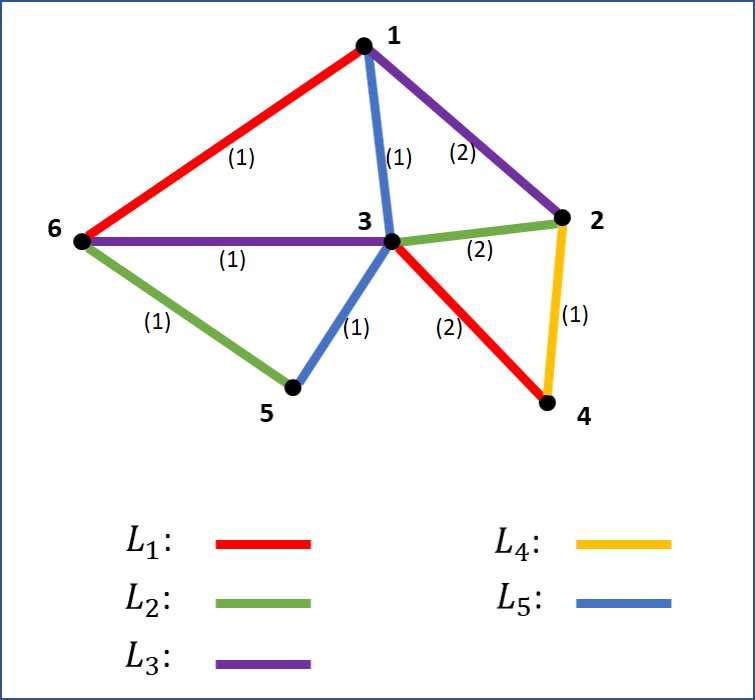}
	\caption{Hedge Graph: each edge has a label $L_j$ and component index (i).}
	\label{fig:hedgegraph}
\end{figure}

The \emph{\textbf{global hedge connectivity problem}} asks for finding the minimum number of hedges, whose edges removal  disconnects the graph. Such set of hedges is called \emph{\textbf{minimum global hedge cut}}. 

Unlike the graph edges adjacency, hedges may have several components. i.e, the span of a hedge is more than $1.$ The hedge component-wise adjacent relationship can be represented as a square matrix $\mathbf{M}$ of order $S,$ where $S$ is the number of the maximum span of all hedges: $S:=\max_{i=1}^{p}(span(H_i)).$ For hedge $H_r$ and hedge $H_t$,  $\mathbf{M}_{ij}=1$ if and only if the the $i$-th component of hedge $H_r$ and the $j$-th component of hedge $H_t$ intersect each other.  The hedge adjacency relationship can be represented as a three-dimensional tensor, with three dimensions as hedge index, component index, and vertex index. Since this paper mainly focuses on combinatorial aspects, its matrix and tensor properties are not further studied. 

Despite the $st$-hedge connectivity problem without hedge overlaps is NP-hard proved by Cai et al.\cite{zhang2011approximation}, the (global) hedge connectivity problem without hedge overlaps is polynomial-time solvable in several special cases, including graphs with bounded treewidth, planar graphs, and instances with bounded label frequency showed by Zhang \cite{zhang2014efficient}. It is also in P when the graph has bounded degree and when for each label and the subgraph induced by the label is connected by Coudert et al.\cite{coudert2007shared}, i.e., the $span(H_i)=1$, this problem is equivalent to hypergraph connectivity problem, which is also known polynomial-time solvable. Xu \cite{xu2018cuts} gave a randomized polynomial-time algorithm for hedgegraphs with constant span. Fox, Panigrahi and Zhang \cite{fox2019minimum} further improved the running time of randomized polynomial-time algorithm for hedgegraphs with constant span and hedgegraph-$k$-cut problem, where $k$ is a constant. The problem is fixed-parameter tractable (FPT) when parameterized by the number $k$ of labels for which the subgraph induced by the label is not connected by Coudert et al.\cite{coudert2016combinatorial}. It is also known hedge connectivity problem without hedge overlaps is quasi-polynomial time solvable by Mohsen, Karger and Panigrahi \cite{ghaffari2017random}. 

There are several other easy verifiable polynomial-time solvable cases: (1). If there is only one label, this graph is $1$- connected; (2). If there is one vertex has label degree $1,$ the hedge graph is $1$-connected; (3). If the number of labels is constant, it is polynomial-time solvable; (4). If there are $m$ labels, where $m$ is the number of edges, this problem is equivalent to the ordinary graph edge connectivity problem. Thus the open case of hedge connectivity without hedge overlaps is the following: when there exists a hedge with span no less than $2,$ all vertices have label degree no less than $2,$ the number of labels is no more than $m-1,$ but this number is not a constant. 

In the following sections, the relationship between hedge connectivity and label degree, hedge adjacency and hedge contraction operations are carefully investigated. 

\section{Label Degree}

Let $d_L(v)$ be the \emph{\textbf{label degree}} of vertex $v\in V$, which is the number of different labels on the edges incident with vertex $v$. Let $\delta_{L}(V)$ and $\Delta_{L}(V)$ be the \emph{\textbf{minimum label degree}} and \emph{\textbf{maximum label degree}} of graph $G$:

$$\delta_{L}(V):= \min_{v_i\in V(G)}\{d_L(v_i)\},$$ 
$$\Delta_{L}(V):= \max_{v_i \in V(G)}\{d_L(v_i)\}.$$

For each hedge $H_i$, let its \emph{\textbf{hedge total label degree}} $total~d_L(V(H_i))$ be the sum of \emph{\textbf{induced label degrees}} of all its vertices, $\delta_L(V(H_i))$ and $\Delta_L(V(H_i))$ be its \emph{\textbf{minimum hedge label degree}} and \emph{\textbf{maximum hedge label degree}}, which are the minimum and maximum value of induced label degrees of all its vertices. 

$$total~d_L(V(H_i)):=\sum_{v_j\in V(H_i)} d_L(v_j),$$ $$\delta_L(V(H_i)):=\min_{v_j\in V(H_i)} d_L(v_j),$$ $$\Delta_L(V(H_i)):=\max_{v_j\in V(H_i)} d_L(v_j).$$

It is obvious that 

$$\delta_{L}(V)= \min_{i} \delta_L(V(H_i)),$$ 
$$\Delta_{L}(V)= \max_{i}\Delta_L(V(H_i)).$$

\begin{theorem}\label{min_label_degree}
	The Global Hedge Connectivity (Minimum Global Hedge Cut) of a graph is at most the minimum label degree of the graph: $\lambda_H(G)\leq \delta_{L}(V).$
\end{theorem}

\begin{proof}
	
	Let $L^*$ be the minimum global hedge cut, if $|L^*|$ is greater than $\delta_{L}(V)$, the hedges adjacent to the vertex of minimum label degree can be seen as a hedge cut, as their removal can disconnect the graph, but their size is smaller than the $|L^*|$, which conflicts $L^*$ is the minimum global hedge cut.
	
\end{proof}

\begin{theorem} \label{max_label_degree} Relabel each hedge with a new label, such that if two hedges adjacent to each other, the new labels of them are different. The new labels form a new label set $L'=\{L'_1, L'_2,..., L'_q\}$, its size is at least the maximum label degree of the original graph: $|L'|\geq \Delta_{L}(V).$ 
\end{theorem}

\begin{proof}
	Assume $d_L(v)=\Delta_{L}(V),$ there are $\Delta_{L}(V)$ hedge edges connected to vertex $v.$ If $|L'|< \Delta_{L}(V),$ according to pigeonhole principle, there exists two hedges whose edges connected to vertex $v$ have the same new label, which contradicts with the relabel rule. 
	
\end{proof}

\begin{figure}
	\centering
	\includegraphics[width=0.6\linewidth]{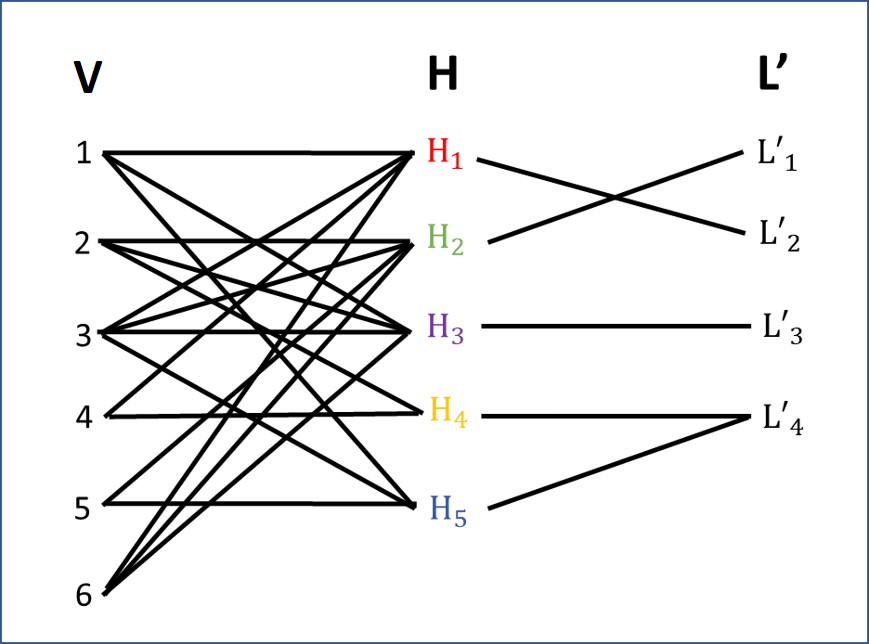}
	\caption{Hedge Relabeling}
	\label{fig:hedgerelabel}
\end{figure}

\section{Hedge Adjacency}

\begin{figure}
	\centering
	\includegraphics[width=0.8\linewidth]{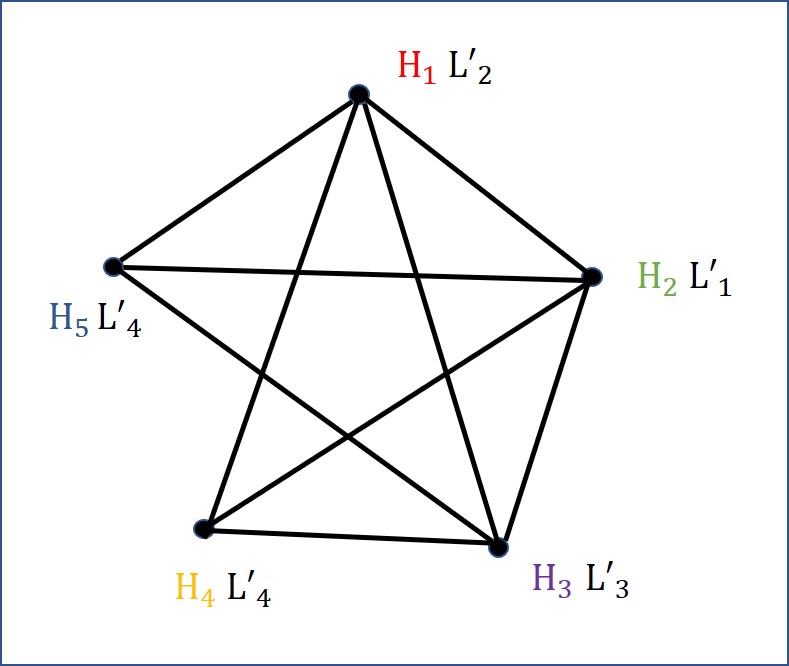}
	\caption{Hedge Adjacency with Relabeling}
	\label{fig:hedgeadcacency}
\end{figure}

\begin{figure}
	\centering
	\includegraphics[width=0.8\linewidth]{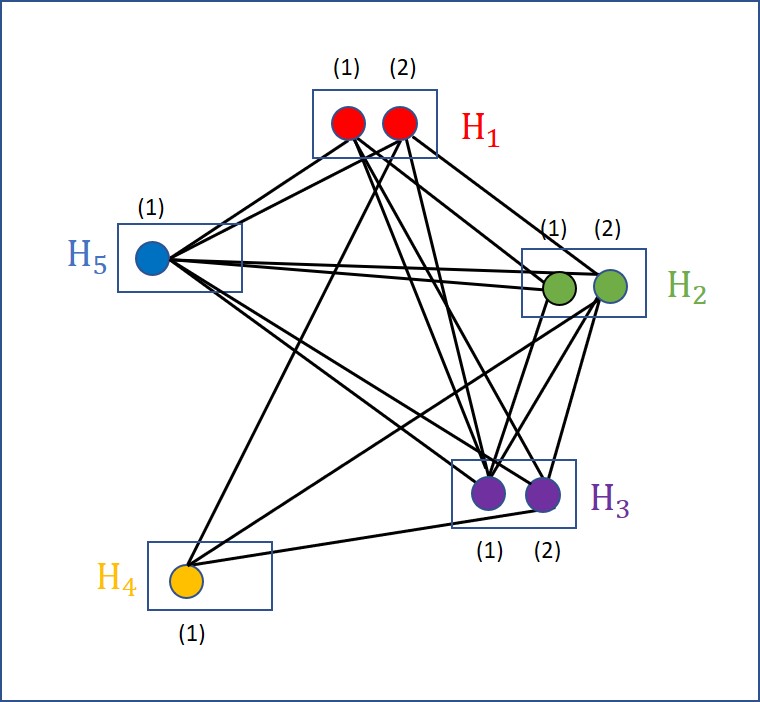}
	\caption{Hedge Component-wise Adjacency }
	\label{fig:hedgeadcacency2}
\end{figure}

Let $d_A(H_i)$ be the \emph{\textbf{hedge adjacency degree}} of hedge $H_i$, which is the number of other hedges adjacent to hedge $H_i$. 

\begin{theorem}
	$d_A(H_i)\leq total~d_L(V(H_i)).$
\end{theorem}

\begin{proof}
	Accoring to the definition, $d_A(H_i)$ is the number of other hedges that hedge $H_i$ adjacent to, which is the union of adjacent hedges of all vertices in $H_i,$ and its size is no more than the the total number of adjacent hedges of all vertices in $H_i.$
	
\end{proof}

\begin{theorem}
	$\max_i~d_A(H_i) \geq \Delta_{L}(V).$ 
\end{theorem}

\begin{proof}
	Assume vertex $v$ has the maximum label degree: $d_L(v)=\Delta_{L}(V),$ and $d_A(H_k)= \max_i~d_A(H_i).$ If $v\in V(H_k),$ $d_A(H_k)$ is the union of adjacent hedges of all vertices in $H_k,$ thus $\max_i~d_A(H_i) \geq \Delta_{L}(V).$ Otherwise, if $v\in V(H_j),$ where $d_A(H_j) \leq d_A(H_k),$ because $d_L(v)\leq d_A(H_j),$ one can get $d_L(v)=\Delta_{L}(V)\leq d_A(H_j) \leq d_A(H_k)= \max_i~d_A(H_i).$
	
\end{proof}

\begin{theorem} \label{max_label_degree} Relabel each hedge with a new label, such that if two hedges adjacent to each other, the new labels of them are different. The new labels form a new label set $L'=\{L'_1, L'_2,..., L'_q\}$, its size is at least the maximum adjacency degree of the original graph.  $|L'|\geq \max_i~d_A(H_i).$ 
\end{theorem}

\begin{proof}
	Draw the hedge adjacency graph $G_A(V_A, E_A)$ such that each vertex in $V_A$ is a hedge of original hedge graph $G,$ two vertices in $V_A$ are connected by an edge if and only if the two corresponding two hedges in the original hedge graph are adjacent to each other. This procedure can be done in polynomial-time. 
	
	Assume $|L'|< \max_i~d_A(H_i),$ and $d_A(H_k)=\max_i~d_A(H_i).$ According to pigeonhole principle, among the hedges adjacent to $H_k,$ there are two hedges have the same label, which contradicts with the relabel rule.
	
\end{proof}

\begin{corollary}
	$\lambda_H(G)\leq\delta_{L}(V)\}\leq\Delta_{L}(V)\}\leq \max_i~d_A(H_i)\leq |L'| .$
\end{corollary}

\begin{lemma} \textbf{(Vizing's Theorem)} \label{vizing}
	The number of labels needed to relabel the hedge graph, such that adjacent hedges in the original graph now have different labels, is either $\max_i~d_A(H_i)$ or $\max_i~d_A(H_i)+1.$
\end{lemma}

\section{Hedge Contraction}

For hedge $H_i,$ $G\diagup H_i$ denotes the graph obtained from $G$ by contracting each edge $e\in H_i$ into a single vertex and deleting resulting loops with label of $H_i.$ Note that, other loops with different labels from $H_i$ caused during contraction of $H_i,$ must be kept, unless they can be cleaned up in the following clean-up process: (1) merge edges with same labels between two vertices; (2) merge loops with sames labels of edges on each vertex. 

In contraction of $H_i$, it is easy to verify that, in graph $G,$ the number of edges is reduced by $|H_i|$, i.e., the number of edges in $H_i$, and the number of vertices is reduced by $(|(V(H_i)|-span(H_i)),$ i.e., the rank of $H_i.$ The rank of $H_i$ is reduced to be zero and the nullity of $H_i$ is reduced by one; the rank of $G$ is reduced by the rank of $H_i$ and the nullity of $G$ is reduced by the nullity of $H_i.$

\begin{theorem}\label{ranksum}
	$$rank(G)=\sum_{i}^{|L|} rank(H_i).$$
\end{theorem}

\begin{proof}
	According to the definition, the rank of graph $G$ is the difference of its number of vertices and the number of its components: $rank(G)=n-1.$ In original graph $G,$ the number of vertices is $n,$ after contacting all hedges, there is only one single vertex left, the number of vertices is reduced by $(n-1)$ in total.  In contraction of $H_i,$ the number of vertices of graph $G$ is reduced by the rank of $H_i,$ i.e., $|V(H_i)|-span(H_i).$ Thus, $$\sum_{i}^{|L|} rank(H_i)=\sum_{i}^{|L|}(|V(H_i)|-span(H_i))=n-1=rank(G).$$
\end{proof}

\begin{theorem}\label{nullsum}
	$$nullity(G)=\sum_{i}^{|L|} nullity(H_i).$$
\end{theorem}

\begin{proof}
	According to the definition, the nullity of graph $G$ is the difference of the number of its edges and its rank: $nullity(G)=m-n+1.$ From Theorem \ref{ranksum}, $\sum_{i}^{|L|} nullity(H_i)= \sum_{i}^{|L|}|H_i|-\sum_{i}^{|L|} rank(H_i)=m-n+1=nullity(G).$
	
\end{proof}

\begin{theorem} \label{vd}
	$$\sum_{i=1}^{|L|}|V(H_i)|=\sum_{j=1}^{n}d_L^G(v_i).$$
\end{theorem}

\begin{proof}
	Since all hedge vertices cover graph $G$ vertices, at each vertex $v_i\in V(G),$ it is covered by at least $d_L^G(v_i)$ times, thus $\sum_{i=1}^{|L|}|V(H_i)|\geq \sum_{j=1}^{n}d_L^G(v_i).$ On the other hand, if one hedge vertex $v_i\in H_j$ is contained in a subset of vertices of $V(G)$, if the subset of vertices of $V(G)$ incident to at least one edge with the same label as hedge $H_j,$ the label degree all all vertices of $G$ is no less than the union of all incident edge sets of all subsets of $V(G),$ thus $\sum_{i=1}^{|L|}|V(H_i)|\leq \sum_{j=1}^{n}d_L^G(v_i).$ Therefore, the equality holds. 
\end{proof}

\begin{theorem}
	$$\sum_{i=1}^{|L|} span(H_i) \leq 2m-n+1.$$
	$$n\delta_L(v_i)-n+1\leq \sum_{i=1}^{|L|} span(H_i) \leq n\Delta_L(v_i)-n+1 .$$
\end{theorem}

\begin{proof}
	According to Theorem \ref{ranksum}, $\sum_{i}^{|L|}(|V(H_i)|-span(H_i))=n-1,$ thus\\ $\sum_{i}^{|L|}span(H_i)= \sum_{i}^{|L|}(|V(H_i)|)-n+1.$ According to Theorem \ref{vd},\\ $\sum_{i=1}^{|L|}|V(H_i)|=\sum_{j=1}^{n}d_L^G(v_i),$ thus $\sum_{i}^{|L|}span(H_i)= \sum_{j=1}^{n}d_L^G(v_i)-n+1.$ It is obvious $\sum_{v_j\in V(G)} d_L(v_j)\leq 2\sum_{i=1}^{|L|}|H_i|=2|E|=2m,$ therefore the first inequality holds. Since $\delta_L(v_i)\leq d_L^G(v_i)\leq \Delta_L(v_i),$ thus the second inequality holds. 
\end{proof}

\begin{theorem} \label{contractv}
	Let $u$ and $v$ be two vertices of hedge graph $G$ with label degrees $d_L(u)$ and $d_L(v)$, and $e_{uv}\in E(G)$,  after contacting the edge $e_{uv}$ between them, the new vertex $w$ replacing them satisfies the following conditions: \\ 
	$\min\{d_L(u), d_L(v)\}-1 \leq \max\{d_L(u), d_L(v)\}-1\leq d_L(w)\leq d_L(u)+d_L(v)-2 \leq 2 \max\{d_L(u), d_L(v)\}-2.$ 
\end{theorem}

\begin{proof}
	Let $L(E(u))$ and $L(E(v))$ be the labels on edges incident with vertices $u$ and $v$. After contacting the edge $e_{uv}$ between $u$ and $v$, the labels on edges adjacent to the new vertex $w$ are the union of $L(E(v))$ and $L(E(v))$ deducing the label of the contracted edge, which number is no more than the sum of each vertex label degree minus 1. $d_L(w)= |(L(E(u))-L(e_{uv}))\cup (L(E(v))-L(e_{uv}))|\leq d_L(u)-1+d_L(v)-1= d_L(u)+d_L(v)-2. $
	
	As the number of labels in this union $(L(E(u))-L(e_{uv}))\cup (L(E(v))-L(e_{uv})$ is no less than the maximum label sets of $L(E(v))-L(e_{uv})$ and $L(E(v))-L(e_{uv})$, thus $\max\{d_L(u), d_L(v)\}-1\leq d_L(w).$ 
\end{proof}

\begin{theorem}\label{ContractMin} 
	$$\delta_{L}(V(G))-1\leq \delta_{L}(V(G\diagup H_i)).$$
\end{theorem}

\begin{proof}
	Assume $v\in V(G)$ has the minimum label degree: $d_L(v)=\delta_L(V(G)).$ (1) If $v\notin V(H_i),$ $d_L(v) \leq \delta_L(V(H_i)),$ after contracting $H_i,$ in worst case, according to Theorem \ref{contractv}, the new created vertices have minimum label degree of $\delta_L(V(H_i))-1.$ Vertices not in $H_i$ keep their original label degrees. Since $d_L(v)\leq\delta_L(V(H_i)),$ therefore $\delta_{L}(V(G))-1\leq \delta_{L}(V(G\diagup H_i)).$
	(2)If $v\in V(H_i),$ $d_L(v)=\delta_L(V(H_i)),$ after contracting $H_i,$ according to Theorem \ref{contractv}, the new created vertices have minimum label degree of $\delta_L(V(H_i))-1=d_L(v)-1.$ Vertices not in $H_i$ keep their original label degrees. Therefore, $\delta_{L}(V(G))-1\leq \delta_{L}(G\diagup H_i).$
\end{proof}

\begin{theorem} \label{contracth}
$$\sum_{v_j\in V(G/H_i)} d_L(v_j)\leq \sum_{v_j\in V(G)} d_L(v_j)-2 rank(H_i).$$
\end{theorem}

\begin{proof}
	According to Theorem \ref{contractv}, after contacting one edge between $u$ and $v$, the new vertex $w$ replaced them has label degree $d_L(w)\leq d_L(u)+d_L(v)-2.$
	Apply contraction to all edges of hedge $H_i,$ for vertices not in $V(H_i),$ their label degrees keep the same, for vertices in $V(H_i),$ their total label degrees are the size of the union of labels on $E(V(H_i))$ deducing the label of $H_i,$ which is no more than the total label degrees of $V(H_i)$ deducing two times the rank of $H_i,$ since the number of vertices reduces by $rank(H_i),$ in each operation of merging two vertices, the number of labels in the new created vertex reduces by at least 2. 
\end{proof}

\begin{theorem}\label{ContractSum}
	$$\sum_{v_j\in V(G)} d_L(v_j) \leq \sum_{v_j\in V(G\diagup H_i)} d_L(v_j)+ \sum_{v_j\in V(H_i)}d_L (v_j)-span(H_i)(\delta(V(G))-1).$$
\end{theorem}

\begin{proof}
	$\sum_{v_j\in V(G)} d_L(v_j)-\sum_{v_j\in V(H_i)}d_L (v_j)$ is the total label degrees of vertices not in $V(H_i),$ after contraction of $H_i,$ those total label degrees do not change. $\sum_{v_j\in V( G\diagup H_i)} d_L(v_j)$ are the total label degrees of all vertices in $G/H_i,$ which contain the total label degrees of vertices not in $V(H_i),$ and the the label degrees of new created vertices after contraction of $H_i.$ According to Theorem \ref{ContractMin}, the total label degrees of new created vertices after contraction of $H_i$ are no less than $span(H_i)(\delta(V(G))-1)$, the above inequality holds. 	
\end{proof}

\begin{theorem}\label{ContractAdjacency}
	$$
	d_A^{G/H_i}(H_j)=\begin{cases}
	d_A^{G}(H_j)+d_A^{G}(H_i)-|L'|+1,~if~H_i~is~adjacent~to~H_j~in~G; \\
	d_A^{G}(H_j),~otherwise.\\
	\end{cases}
	$$
	
	where $|L'|$ is either $\max_i~d_A^G(H_i)$ or $\max_i~d_A^G(H_i)+1.$
	
\end{theorem}

\begin{proof}
	If $H_i$ is adjacent to $H_j$ in $G$, after contracting of $H_i,$ $H_i$ is deleted from the adjacency hedge list of $H_j,$ but hedges originally adjacent to $H_i$ but not $H_j$ now become adjacent to $H_j$ via the contraction operation. In the relabel process, hedges have different new labels if they are adjacent to each other, thus there are $(|L'|-2)$ number of hedges are adjacent to both $H_i$ and $H_j,$ the number of hedges adjacent to $H_i$ but not $H_j$ is $(d_A^{G}(H_i)-|L'|+2).$ Therefore, $d_A^{G/H_i}(H_j)=d_A^{G}(H_j)-1+(d_A^{G}(H_i)-|L'|+2).$
	
	According to Lemma \ref{vizing}, $|L'|$ is either $\max_i~d_A^G(H_i)$ or $\max_i~d_A^G(H_i)+1.$
	
	If $H_i$ is not adjacent to $H_j$ in $G$, after contraction of $H_i,$ $d_A^G(H_j)$ does not change. 
\end{proof}

\paragraph{\textbf{Acknowledgments}}

Rupei Xu would like to express her sincere gratitudes to Guoli Ding and Jie Han for their helpful discussions.

\bibliographystyle{plain}
\bibliography{example}

\end{document}